\documentclass[JHEP,12pt]{article}
\usepackage{subfigure}
\usepackage{amssymb,amsmath,slashed}
\usepackage{braket}
\usepackage{jheppub}
\usepackage{mathrsfs}
\usepackage{amsthm}
\newtheorem{thm}{Theorem}[section]

\theoremstyle{remark}
\newtheorem{defn}[thm]{Definition}
\newtheorem{rem}[thm]{Remark}
\newtheorem{conj}[thm]{Conjecture}
\newtheorem{conv}[thm]{Convention}
\DeclareMathOperator{\tr }{tr}

\DeclareMathOperator{\A}{Area}

\renewcommand{\S}{S_{\rm gen}}
\usepackage{bbold}
\usepackage{comment}

\author[1]{Raphael Bousso}
\author[2]{and Arvin Shahbazi-Moghaddam}

\affiliation[1]{Berkeley Center for Theoretical Physics and Department of Physics,\\
University of California, Berkeley, CA 94720, USA} 

\affiliation[2]{Stanford Institute for Theoretical Physics,\\ Stanford University, Stanford, CA 94305, USA}

\emailAdd{bousso@berkeley.edu}
\emailAdd{arvinshm@gmail.com}

\title{Quantum Singularities}

\abstract{Two spatial regions $B$ and $R$ are \emph{hyperentangled} if the generalized entropy satisfies $\S^{B\cup R}<\S^R$. If in addition all future (or all past) directed inward null shape deformations of $B$ decrease $\S^{B\cup R}$, then we show that the causal development of $B$, with $R$ held fixed, must be incomplete. This result eliminates the Null Energy Condition from the assumptions of a recently proven singularity theorem. Instead, we assume a quantum version of the Bousso bound. 

Taking $R$ to contain the Hawking radiation after the Page time, our theorem predicts a singularity in the \emph{past} causal development of the black hole interior. This is surprising because the classical spacetime is nonsingular in the past. However, one finds that Cauchy slices that are required to contain $R$ do not remain in the semiclassical regime. The \emph{quantum singularities} predicted by our theorem are an obstruction to further semiclassical evolution, generalizing the singularities of classical general relativity.}

\begin{document}
\maketitle

\section{Introduction}
\label{intro}

A spacetime $M$ is singular if it contains an incomplete timelike or null geodesic~\cite{wald2010general} (an inextendible geodesic of finite affine length). Physically relevant examples include the past singularity in certain cosmological solution---the ``big bang''---and the future singularity that terminates time evolution inside a Schwarzschild black hole. 

Singularities are generic in classical general relativity. A theorem by Penrose~\cite{Penrose:1964wq} guarantees that at least one of the null geodesics orthogonal to a trapped surface is incomplete. A surface is trapped if both sets of future-directed orthogonal null geodesics have negative expansion.

Penrose's theorem requires two crucial assumptions about the spacetime: $M$ must admit a noncompact Cauchy surface; and $M$ must satisfy the Null Curvature Condition, $R_{ab} k^a k^b\geq 0$, where $R_{ab}$ is the Ricci tensor and $k^a$ is any null vector. 

A recent result~\cite{Bousso:2022cun} has established a connection between singularities and quantum information: the noncompactness assumption can be eliminated from Penrose's theorem, if instead the spacetime is assumed to satisfy the Bousso bound~\cite{Bousso:1999xy} on the entropy of matter. 

\begin{figure}
\begin{center}
\includegraphics[width=0.9\textwidth]{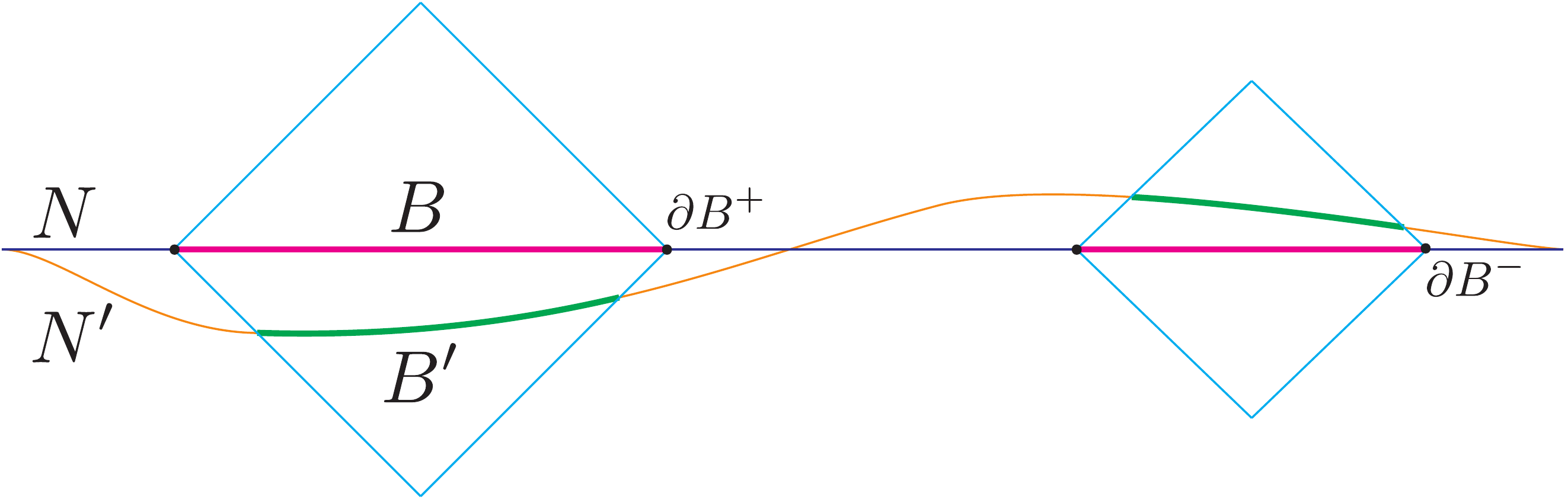}
\end{center}
\caption{\label{fig-bound}The quantum Bousso bound (Conjecture \ref{qbb}): if the quantum expansion at $\partial B$ in the direction of $B'$ is nonpositive then $\S^{B'}\leq \S^B$.}
\end{figure}
The Null Curvature Condition can also be eliminated. This is important, because it it is known not to hold in Nature. By Einstein's equation, it is equivalent to the Null Energy Condition, that $T_{ab} k^a k^b\geq 0$, where $T_{ab}$ is the stress tensor. Any relativistic quantum field theory, such as the Standard Model, contains states in which the expectation value of the stress tensor, $\braket{T_{ab}}$, violates this condition~\cite{Epstein:1965zza}. Wall~\cite{Wall:2010jtc} eliminated the Null Curvature Condition from Penrose's theorem, by assuming instead that the Generalized Second Law (GSL) holds in $M$. The GSL is the statement that the generalized entropy---the sum of horizon area and von Neumann entropy of the matter fields outside a causal horizon---cannot decrease. A causal horizon is the boundary of the past of a timelike or null curve of infinite affine length; examples include black hole, Rindler, and de Sitter horizons. Unlike for the NCC, there is no known counterexample to the GSL. There is considerable evidence for its validity, and it has been proven to hold on Killing horizons~\cite{Wall:2011hj}.

In this paper, we combine the advances of Refs.~\cite{Bousso:2022cun, Wall:2010jtc}, using a single assumption, a quantum refinement of the Bousso bound~\cite{Bousso:2015mna}. This bound says that if the generalized entropy outside a Cauchy-splitting null hypersurface $L$ is decreasing towards the future (resp.\ past) at some moment of time, then it must be lower at all future (past) times. See Fig.~\ref{fig-bound}; and see Conj.~\ref{qbb} below for a more precise statement. The quantum Bousso bound implies the GSL as a special case.

\begin{figure}
\begin{center}
\includegraphics[width=0.85\textwidth]{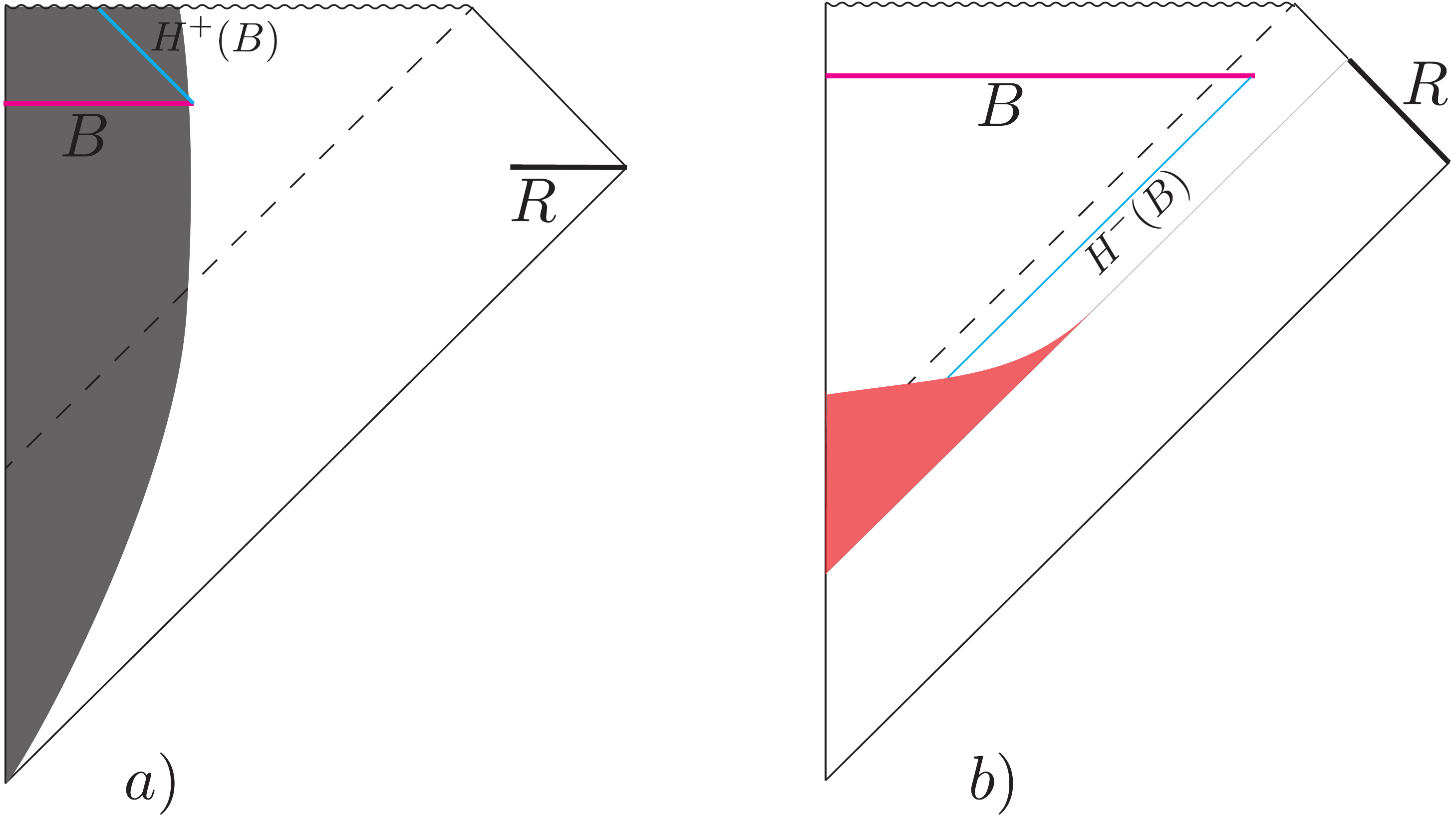}
\end{center}
\caption{\label{fig-theorem}Quantum singularity theorem for a hyperentangled region $B$. (a) Collapsing star entangled with a distant reference system. At late times the entanglement entropy exceeds the surface area of the star and the theorem predicts a future singularity. This is the (ordinary) singularity of the black hole. (b) Evaporating black hole after the Page time. The black hole interior $B$ and distant Hawking radiation $R$ are hyperentangled, and the theorem predicts both a future and a past singularity. The latter is a quantum singularity: it arises from the breakdown of semiclassical evolution in the red region when $R$ is held fixed.}
\end{figure}
We will prove that the quantum Bousso bound implies a singularity theorem for certain hyperentangled regions, Theorem~\ref{qst} below. We call a spatial region $B$ hyperentangled if 
\begin{equation}
    \S^{B\cup R}<\S^R
\end{equation}
for some spacelike-separated region or external system $R$, where $\S$ is the generalized entropy. Suppose that, in addition, $B\cup R$ has negative inward quantum expansion at $\partial B$; that is, $\S^{B\cup R}$ decreases under shape deformations of $B$ along a past-directed (or a future-directed) ingoing null congruence. See Fig.~\ref{fig-theorem}. 

Under these assumptions, we prove that at least one null geodesic in the congruence is incomplete, \emph{in any spacetime obtained from semiclassical evolution on Cauchy surfaces that all contain $R$.}

We shall see through the study of examples that the quantum Bousso bound, and indeed the GSL, are surprisingly restrictive when Cauchy evolution is limited to slices containing $R$. As a result, spacetimes that satisfy this bound admit a novel, $R$-dependent notion of singularity. Our theorem captures such singularities. We believe that the notion of $R$-dependent singularities in semiclassical gravity has not been discussed in the literature, so we will do so now. 

\begin{figure}
\begin{center}
\includegraphics[width=0.4\textwidth]{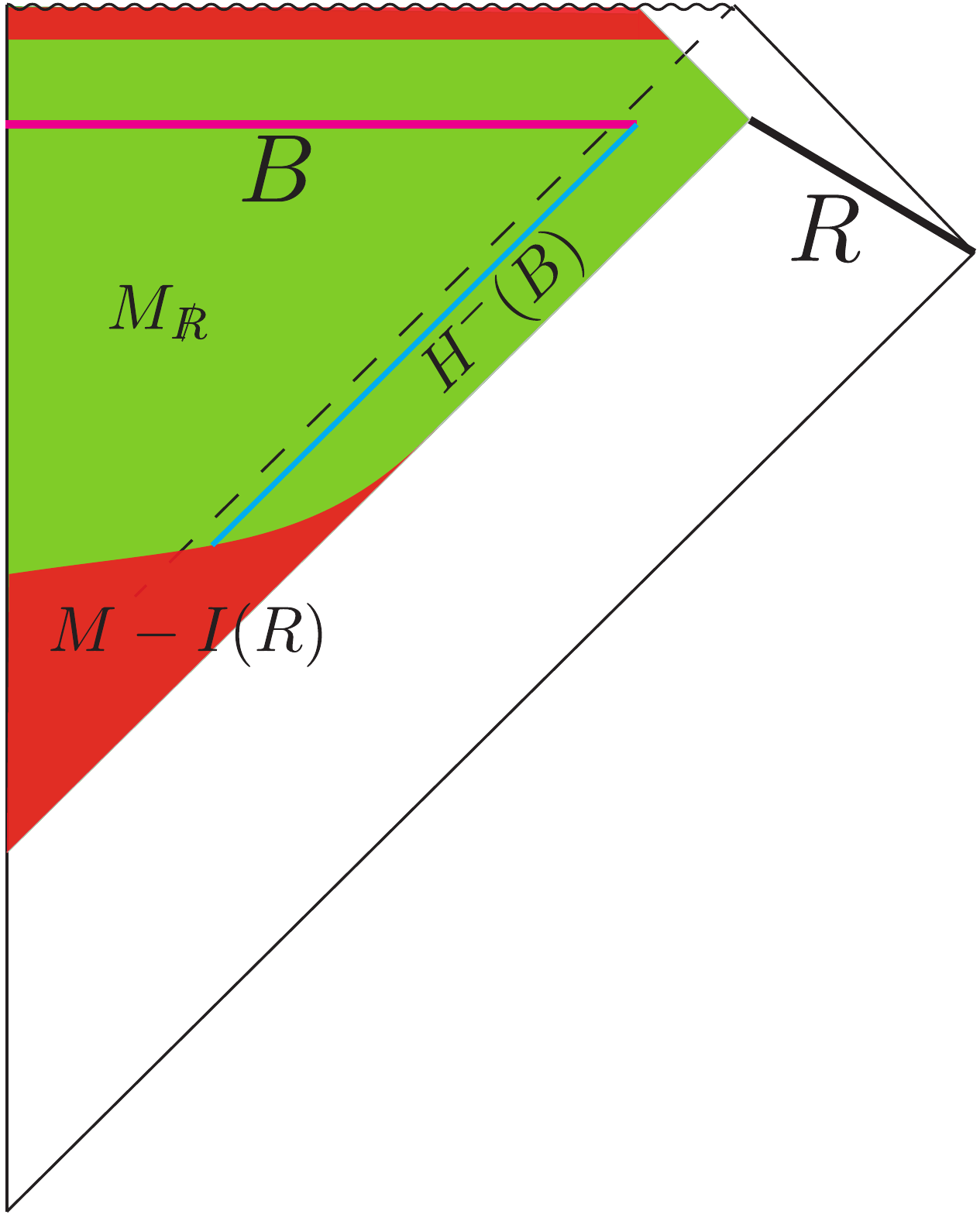}
\end{center}
\caption{\label{fig-mrir}$M_{\slashed{R}}$ (green) is a spacetime covered by nice slices that contain $R$. In general this semiclassical region is a proper subset of the region $M-I(R)$ (red+green) covered classically by Cauchy slices that contain $R$. Generators of $H^\pm(B)$ may be incomplete in $M_{\slashed{R}}$ even if they would be complete in $M-I(R)$. This reflects a real obstruction to further semiclassical evolution which we call a \emph{quantum singularity}.}
\end{figure}
In classical General Relativity, the inclusion of $R$ in all Cauchy slices would be a trivial restriction for the purposes of our theorem. As shown in Fig.~\ref{fig-mrir}, slices that contain $R$ foliate $M-I(R)$, where $I(R)$ is the union of the chronological past and future of $R$. The Cauchy horizon $H^\pm(B)$ is a subset of $M-I(R)$, so if it contains an incomplete geodesic, then so does $M$.


However, in semiclassical gravity, Cauchy slices must be ``nice.'' That is, the data on each slice must be compatible the validity of an effective field theory description with a cutoff below the Planck scale.\footnote{We are not aware of a first-principles derivation of the precise inequalities that ``niceness'' implies for scalar quantities extracted from the initial data on a slice. A plausible set of criteria was outlined in Ref.~\cite{Lowe:1995ac,Polchinski:1995ta}. To understand the quantum singularities predicted by our theorem in the examples we study here, we need only require a rather weak niceness condition on the trace of the extrinsic curvature; see Sec.~\ref{evap}.} The boundary of $M-I(R)$ is null; therefore, it contains distinct points with zero spatial distance, and nice slices that contain $R$ cannot approach it arbitrarily closely; see Fig.~\ref{fig-mrir}. 

Thus, only a subset $M_{\slashed{R}}\subset M-I(R)$ can be obtained by semiclassical evolution along Cauchy surfaces that all contain $R$. Hence it is possible for $H^\pm(B)$ to contain a geodesic that is incomplete in any semiclassically allowed spacetime $M_{\slashed{R}}$, even though it would be complete in the (larger but semiclassically unattainable) $M-I(R)$. We call a singularity that arises in this manner a \emph{quantum singularity}. 

\paragraph{Outline} In Sec.~\ref{semicl}, we define a semiclassical spacetime as a globally hyperbolic manifold $M$ with metric $g_{ab}$ and quantum state $\rho$ whose stress tensor expectation value satisfies the Einstein equation, $G_{ab}= 8\pi G\braket{T_{ab}}$. In addition, $M$ must admit a ``nice slicing.'' That is, time evolution must be consistent with the validity of an effective field theory description, with an ultraviolet cutoff $\Lambda\ll M_P$, where $M_P$ is the Planck mass. We formulate the main assumption of our theorem, the quantum Bousso bound, and we clarify that it applies to nice slices only.\footnote{We also comment on Ref.~\cite{Rolph:2022csa}, which asserts a different regime of validity of the quantum Bousso bound (and, implicitly, of the GSL), and which arrives at conclusions different from ours.}

After proving our theorem in Sec.~\ref{sec-qst}, we discuss two nontrivial applications. In Sec.~\ref{pastofb}, we consider an evaporating black hole formed from regular initial conditions. We apply our theorem to the black hole interior $B$ after the Page time; $R$ is the Hawking radiation emitted so far. In this case a singularity is predicted along the black hole horizon to the \emph{past} of $B$. 

This is a quantum singularity. It appears because we are holding $R$ fixed, thus excluding the region $R$ from participating in the semiclassical evolution. On nice slices that contain the Hawking radiation, the horizon cannot be evolved far into the past. Hence the horizon generators are incomplete in $M_{\slashed{R}}$. The semiclassically allowed spacetime $M_{\slashed{R}}$ is substantially smaller than $M-I(R)$, which contains the entire black hole horizon.

In Sec.~\ref{inner}, we consider the Kerr-Newman solution in thermal equilibrium with a bath. We again consider the black hole interior $B$ at a sufficiently late time, when it is entangled with distant radiation $R$. Our theorem predicts a singularity before the inner horizon. In the classical Kerr-Newman solution, the region between the inner and outer horizons is regular, so this conclusion is of some interest. We discuss its potential relevance to the strong cosmic censorship hypothesis. 

\section{Semiclassical Gravity}
\label{semicl}

\subsection{Causal Structure}
\label{causal}

\begin{conv}\label{Mconv}
Everywhere below, $M$ will denote a time-orientable globally hyperbolic spacetime. ($M$ may be extendible.) We use an overdot to represent the boundary of a subset of $M$.
\end{conv}
\begin{conv}\label{Bconv}
Everywhere below, $B$ will denote a closed subset of a Cauchy slice $N$ of $M$, such that $\partial B$ is a compact codimension 2 submanifold of $M$ and $B-\partial B\neq \varnothing$. Here $\partial B$ denotes the boundary of $B$ in the induced topology of $N$. 
\end{conv}

\begin{defn}\label{ijdef}
The \emph{chronological} and \emph{causal future} and \emph{past}, $I^\pm(K)$ and $J^\pm(K)$, of any set $K\subset M$ are defined as in Wald~\cite{wald2010general}. For $K=\set{p}$, we drop the set brackets. Key consequences of these definitions include: $p\notin I^+(p)$ but $p\in J^+(p)$, and $I^+(K)$ is open.
\end{defn}

\begin{defn}\label{idef}
For any set $K\subset M$, we define its \emph{domain of influence} as the union of $K$ and all points that can be reached by a timelike curve from $K$:
$I(K) \equiv I^+(K) \cup I^-(K) \cup K$.
\end{defn}

\begin{defn}\label{ddef}
For any closed achronal set $K\subset M$, the \emph{future domain of dependence}, $D^+(K)$, is the set of points $p$ such that every past-inextendible causal curve through $p$ must intersect $K$. The \emph{past domain of dependence}, $D^-(K)$, is defined analogously. The \emph{domain of dependence} is $D(K)\equiv D^+(K)\cup D^-(K)$. The \emph{future Cauchy horizon} is $H^+(K)\equiv \overline{D^+(K)}-I^-[D^+(K)]$.
\end{defn}

\begin{defn}
Let $M$ be a spacetime with Cauchy surface $N$. ($N$ or $M$ may be extendible.) Let $B\subset \Sigma$. We say that $B$ is \emph{future singular in $M$} if its Cauchy horizon $H^+(B)$ contains an incomplete geodesic; \emph{i.e.}, a geodesic that is future-inextendible in $M$ but of finite affine length. Otherwise, we call $B$ \emph{future complete in $M$}.
\end{defn}

\subsection{Kinematics}

\begin{defn}
A \emph{nice global slice} is an inextendible Cauchy surface $\Sigma$ whose intrinsic and extrinsic geometry and quantum state $\rho(\Sigma,\Lambda)$ can be fully described using a cutoff scale $\Lambda\ll M_P$. In particular, curvature scalars and energy densities that can be constructed from the normal vector to $\Sigma$ must be sub-Planckian.
\end{defn}


\begin{defn}\label{nice}
Let $\Sigma$ be a nice global slice of $M$, with associated cutoff scale $\Lambda\ll M_P$. Let $B$ be defined as in Convention~\ref{Bconv}. If the intrinsic and extrinsic geometry of $\partial B$ is well resolved at the cutoff $\Lambda$ (that is, when $\partial B$ is put on a lattice with characteristic scale $\Lambda^{-1}$), then we call $B$ a \emph{nice partial slice}, and the quantum state on $B$ is defined as
\begin{equation}
    \rho(B,\Lambda) = \tr_{\Sigma-B} \rho(\Sigma,\Lambda)~,
\end{equation}
\end{defn}

\begin{defn}
We call $N$ a \emph{nice slice} if $N$ is a nice global slice or a nice partial slice.
\end{defn}
\begin{defn}\label{sgendef}
Let $N$ be a nice slice with associated cutoff scale $\Lambda$. The \emph{generalized entropy} of $N$ is
\begin{equation}\label{sgendefeq}
    \S^N=\frac{\A(\partial N)}{4G(\Lambda)\hbar}+\ldots+S(N,\Lambda)~,
\end{equation}
where 
$G(\Lambda)$ is the effective Newton constant at the scale $\Lambda$, and
\begin{equation}
    S(N,\Lambda) = -\tr_N \rho(N,\Lambda) \log \rho(N,\Lambda)
\end{equation}
is the von Neumann entropy of the quantum fields on $N$ at the scale $\Lambda$.
The area term is the leading gravitational counterterm. The subleading gravitational counterterms are indicated by ``$\ldots$''; see Ref.~\cite{Dong:2013qoa} for details.
\end{defn}

\begin{rem}
Niceness of $N$ is required in the above definition since otherwise $G(\Lambda)$ is not operationally defined, for example as the effective gravitational coupling in a scattering process. The boundary of a nice global slice ($N=\Sigma$) vanishes. The boundary of a nice partial slice ($N=B$) is understood to be defined in a completion $\Sigma\supset N$, where $\Sigma$ is a nice global slice. Neither the boundary area nor the generalized entropy of $N$ will depend on the choice of completion.
\end{rem}

\begin{conj}
The generalized entropy is cutoff-independent, in the following sense. Suppose that the slice $N$ is nice with respect to two different scales $\Lambda$ and $\Lambda'$. Under $\Lambda\to\Lambda'$, both terms in Eq.~\eqref{sgendefeq} will change, but their sum will not. For references supporting this claim, see the Appendix of Ref.~\cite{Bousso:2015mna}.
\end{conj}


\subsection{Dynamics}

\begin{defn}\label{slicing}
A \emph{slicing} of the spacetime $(M,g)$ is a continuous map from an open interval to achronal subsets of $M$, $t\to N(t)$, such that every point in $M$ is contained in at least one $N(t)$, and $N(t')\subset J^+[N(t)]$ for $t'>t$. [Thus, a slicing is not a foliation. Along a timelike curve $\gamma$, the proper time of $\gamma\cap N(t)$ increases monotonically with $t$, but not strictly so.]
\end{defn}

\begin{defn}
A \emph{Cauchy slicing} of $(M,g)$ is a slicing such that each $N(t)$ is a Cauchy surface of $M$. 
\end{defn}
\begin{defn}
A \emph{nice Cauchy slicing} of $(M,g)$ is a Cauchy slicing such that each $N(t)$ is a nice slice with the same associated cutoff $\Lambda\ll M_P$. A collection of nice Cauchy slicings with cutoff $\Lambda$ will be denoted ${\cal S}_\Lambda$.
\end{defn}
\begin{defn}
A \emph{semiclassical spacetime} is a quadruplet $(M,g,{\cal S}_\Lambda,\rho$). Here $M$ is a globally hyperbolic manifold with metric $g$. ${\cal S}_\Lambda$ is a nonempty set of nice Cauchy slicings. For each slicing, $\rho(N(t),\Lambda)$ solves the Schrödinger equation of the quantum fields. The expectation values of local operators do not depend on the slicing. Finally,
\begin{equation}
    G_{ab} = 8\pi G(\Lambda) \braket{T_{ab}}+\ldots~,
\end{equation}
where $G_{ab}$ is the Einstein tensor computed from $g$, $T_{ab}$ is the stress tensor (viewed as an operator), and ``$\ldots$'' stands for higher-curvature corrections.
\end{defn}

\begin{rem}
The above definition ignores gravitons; this can be justified by taking the number of matter fields to be large. The Cauchy slices of $M$ may be partial and hence extendible; and in any case $M$ may be extendible.
\end{rem}

\begin{rem}\label{remark-1}
Given a nice slice $N(0)$, one can solve the quantum field theory and Einstein's equation iteratively in $G\hbar$, in some open neighborhood of $N(0)$, and thus generate a semiclassical spacetime.
\end{rem}

\begin{defn}\label{reduction}
Let $N$ be a nice Cauchy slice in a semiclassical spacetime $M$, and let $R\subset N$, $R\neq N$. A semiclassical spacetime with nice slicing $N_{\slashed{R}}(t)$ will be called a \emph{reduction of $M$ by $R$} and will be denoted $M_{\slashed{R}}$, if for every $t$, $N_{\slashed{R}}(t)\cup R$ is a nice slice of $M$. (See Fig.~\ref{fig-mrir} for an example.)
\end{defn}

\begin{defn}\label{qedef}
Let $M$ be a semiclassical spacetime, and let $B\subset M$ be a nice partial slice. The future-directed ingoing \emph{quantum expansion} of $B$ at $y\in \partial B$ is the rate of change of the generalized entropy under a shape deformation of $B$ along the ingoing future-directed null congruence orthogonal to $\partial B$:
\begin{equation}\label{qedefeq}
    \Theta_+^B(y) = \frac{4G\hbar}{\sqrt{h(y)}}\left.\frac{\delta \S[V]}{\delta V(y)}\right|_{\partial B}~.
\end{equation}
Here $h$ is the area element of the induced metric on $\partial B$. The functional derivative is taken with respect to the affine parameter $V(y)$ along the congruence that specifies the location of cuts of the congruence such as $\partial B$. 

The past-directed ingoing quantum expansion is defined analogously. Outgoing quantum expansions are related to the ingoing ones by a change of sign and exchange of past and future.
\end{defn}

\begin{rem}
The functional derivative in Eq.~\eqref{qedefeq} is an idealization that suppresses the cutoff $\Lambda$. The quantum expansion is well-defined only if one of the nice Cauchy slicings of $M$ contains slices that contain $B$ and its shape deformation. In particular, this excludes deformations whose transverse support near $y$ is localized to better than $\Lambda^{-1}$~\cite{Leichenauer:2017bmc}.
\end{rem}

\begin{conj}[Quantum Bousso Bound]\label{qbb}
Let $M$ be a semiclassical spacetime, and let $N$ and $N'$ be slices in one of the nice Cauchy slicings of $M$. Let $B\subset N$ be a nice partial slice, and let $\partial B^+$ ($\partial B^-$) be the subset of $\partial B$ with positive future (past)-directed inward quantum expansion. Let $B'=D(B)\cap N'$. If $N'\cap I^+(\partial B^+)=N'\cap I^-(\partial B^-)=\varnothing$, then

\begin{equation}
 \S^{B'}\leq \S^B~. 
\end{equation}
\end{conj}

\begin{rem}
The above conjecture was originally obtained as a consequence of the Quantum Focusing Conjecture~\cite{Bousso:2015mna}. However, its derivation was somewhat heuristic and omitted a careful regularization of points where null generators leave $\dot I(C)$. Here we will assume Conjecture~\ref{qbb} directly. The original Bousso bound~\cite{Bousso:1999xy} follows in the limit where $\Theta$ is well approximated by the classical expansion and $M$ satisfies the Null Curvature Condition.
\end{rem}

\begin{rem}
(Note added.) Compared to our Conj.~\ref{qbb}, Sec.~2 of~\cite{Rolph:2022csa} imposes the additional requirement that $B'\neq \varnothing$, and more strongly that $N'$ must intersect every connected component of $D(B)$. We believe this restriction is too strong and also unnecessary. An important manifestation of the GSL is the fact that the generalized entropy outside the horizon of a black hole is larger than the the (ordinary) entropy before the black hole has formed. This key feature follows from the quantum Bousso bound only if $B'=\varnothing$ is allowed. On the other hand, Ref.~\cite{Rolph:2022csa} does not restrict the application of the bound and of the GSL to nice slices. We believe that this is too permissive, even when combined with the restriction to nonempty $B'$ advocated in Ref.~\cite{Rolph:2022csa}. As we discuss at the end of Sec.~\ref{pastofb}, the generalized entropy becomes negative and stops making sense on slices allowed by this set of criteria.\footnote{With Rolph's conditions on the GSL and the quantum Bousso bound, the proof of the ``island finder'' theorem~\cite{Bousso:2021sji} would indeed have a loophole as claimed in Ref.~\cite{Rolph:2022csa}. With ours, it does not. In the case of concern to Ref.~\cite{Rolph:2022csa}, Conj.~\ref{qbb} would be violated, so a (possibly quantum) singularity must form. (In the maximin formalism~\cite{Wall:2012uf}, it is necessary to assume that the maximin slice is repelled by singularities. We propose that this feature extends to quantum singularities.)}
\end{rem}

\section{Singularity Theorem}
\label{sec-qst}

\begin{thm}[Singularity Theorem for Hyperentangled Regions]
\label{qst}
Let $M$ be a semiclassical spacetime with nice Cauchy slice $N$. Let the disjoint union $B\cup R \subset N$ be a nice slice, with $\partial B$ compact. Suppose that the future-directed inward quantum expansion of $B\cup R$ is negative everywhere on $\partial B$:
\begin{equation}\label{qsta1}
    \Theta_+^{B\cup R}(y)<0~~\mbox{for~all}~~ y\in \partial B~.
\end{equation}
Suppose moreover that $B$ is hyperentangled with $R$, that is:
\begin{equation}\label{qsta2}
    \S^{BR}<\S^R~.
\end{equation}
Let $M_{\slashed{R}}$ be a reduction of $M$ by $R$ (see Def.~\ref{reduction}). Assuming Conjecture~\ref{qbb} (Quantum Bousso Bound), $B$ is future singular in $M_{\slashed{R}}$, \emph{i.e.}, $H^+(B)\cap M_{\slashed{R}}$ contains an incomplete null geodesic. 
\end{thm}
\begin{proof}
The Cauchy horizon $H^+(B)$ is topologically the direct product of $\partial B$ with the future-inward directed null geodesics orthogonal to $\partial B$, up to possible identifications of their endpoints on $H^+(B)$. By Eq.~\eqref{qsta1}, no null geodesic can remain on $H^+(B)$ for infinite affine time~\cite{Wall:2011hj}.\footnote{For if such a geodesic $\gamma$ did exist, then $\dot I^-(\gamma)$ would be a causal horizon, and by the Generalized Second Law, $\Theta_+^{(I^-(\gamma)\cap N)\cup R}\geq 0$. By construction, $I^-(\gamma)\subset B$ and $\partial B$ touches $\dot I^-(\gamma)$ at $p=\gamma\cap \partial B$. 
By Theorem 3 of Ref.~\cite{Wall:2010jtc}, $\Theta_+^{B\cup R}(p)\geq \Theta_+^{(I^-(\gamma)\cap N)\cup R}(p)\geq 0$,
which contradicts Eq.~\eqref{qsta1}.} Assuming for contradiction that $B$ is future complete in $M_{\slashed{R}}$, it follows that $H^+(B)$ contains the endpoints of all of its generators. Compactness of $\partial B$ then implies that $H^+(B)$ is compact.

Let $N_{\slashed{R}}(0) = N-R$, and assume for contradiction that $N_{\slashed{R}}(t) \cap H^+(B) \neq \emptyset$ for all $t\geq0$. Let $t_n$ be a monotonically increasing sequence that converges to the upper bound of the time interval for which the slicing $N_{\slashed{R}}(t)$ is defined (or diverges to $\infty$ if there is no upper bound), and let $x_n\in N_{\slashed{R}}(t_n)\cap H^+(B)$.\footnote{This step invokes the axiom of choice; perhaps this can be eliminated.} By compactness of $H^+(B)$, the sequence $x_n$ has an accumulation point $p\in H^+(B)$. Let $q\in I^+(p)$ and let $N_{\slashed{R}}(t_q)$ be a slice that contains $q$. Because a slicing moves forward in time monotonically by Def.~\ref{slicing}, there exists a small neighborhood $O(p)$ that no slice with $t\geq t_q$ can intersect. This contradicts the fact that $p$ is an accumulation point. 

Therefore $M_{\slashed{R}}$ admits a nice slice such that $N_{\slashed{R}}(t_{\rm above}) \cap H^+(B) = \emptyset$, $t_{\rm above}>0$, and
by Def.~\ref{reduction}, $M$ admits a nice slice 
\begin{equation}
    N'\equiv N_{\slashed{R}}(t_{\rm above})\cup R
\end{equation}
that contains $R$ and fails to intersect $D(B)$. $N'$ satisfies the assumptions of the quantum Bousso bound as applied to $B\cup R\subset N$. (In Conjecture~\ref{qbb}, substitute $B\to B\cup R$.) Hence
\begin{equation}
    \S^R\leq \S^{BR}~,
\end{equation}
which contradicts Eq.~\eqref{qsta2}. Hence $B$ must be future singular in $M_{\slashed{R}}$.
\end{proof}

\begin{rem}
Note that the assumption \eqref{qsta2} cannot be satisfied if $R=\varnothing$, so any nontrivial application of the theorem requires a nonempty choice of $R$. However, $R$ can be arbitrarily far from $B$. After a straightforward adaptation of the relevant definitions, $R$ can even be treated as a nongravitating quantum system that is external to the spacetime. In that case, $M_{\slashed{R}}$ can be an inextendible spacetime.
\end{rem}

\begin{rem}
The singularity theorem for hyperentropic regions~\cite{Bousso:2022cun} emerges in the limit as $\hbar\to 0$. In this limit, $\Theta\to\theta$, so the quantum expansion is well approximated by the classical expansion. Moreover, nice slices will cover all of $M-I(R)$ in this limit. If the entropy in $B$ is not purified by some disjoint region $R$ then an appropriate external purification can be added.
\end{rem}


\begin{rem}
By the previous remark, Theorem~\ref{qst} applies to all of the examples discussed in Ref.~\cite{Bousso:2022cun}, which include several settings where Penrose's theorem would not apply. In all cases we must first introduce an external purification $R$ of the matter entropy in $B_i$. In the following section, we study examples of singularities predicted by Theorem~\ref{qst} that have no classical analogue.
\end{rem}

\section{Hyperentangled Black Holes}
\label{evap}

The conditions of Theorem~\ref{qst} can be satisfied by choosing $B$ to be a region in an evaporating black hole after the Page time, with $R$ a region containing the Hawking radiation. One can arrange that both quantum expansions are negative, $\Theta^{B\cup R}_\pm |_B<0$, so the Theorem predicts a singularity both along $H^+(B)$ and along $H^-(B)$ when $R$ is held fixed.

Let us discuss this in more detail. A slice of the black hole interior after the Page time is by definition hyperentangled with the Hawking radiation emitted so far; let $R$ be the region containing this radiation. By picking $B$ to be the interior of a sphere which is slightly outside of the horizon one can arrange $\Theta^{B\cup R}_\pm |_B<0$.

Alternatively, one can obtain a region $B$ with these properties by deforming the island $I$~\cite{Penington:2019npb,Almheiri:2019psf} associated to $R$. By quantum maximin~\cite{Akers:2019lzs}, islands generically satisfy $\partial_\ell \Theta_k = \partial_k \Theta_{\ell}<0$, where $k$ and $\ell$ are null future-directed orthogonal vectors fields on $\partial I$ outward and inward to $I$ respectively. Therefore, by slightly deforming the island in the future-outward and past-outward null directions, one obtains a hyperentangled region which satisfies the conditions of Theorem~\ref{qst} both in the future and past directions. One can also use this method in an eternal black hole coupled to a bath~\cite{Almheiri:2019yqk} to find a region $B$ with these properties.

Naively, both the future and the past applications of Theorem~\ref{qst} to such a region are quite puzzling. Schwarzschild black holes have a singularity along $H^+(B)$; but for a Schwarzschild black hole formed from regular initial conditions, $H^-(B)$ is complete by construction. At the classical level, even $H^+(B)$ is complete when charge or angular momentum is present. 
Small classical perturbations are believed to produce a spacelike singularity before the inner horizon, but the conditions of our theorem are satisfied in the unperturbed Kerr-Newman solution (see Fig. \ref{fig-RN1}). 

However, the spacetime $M_{\slashed{R}}$ covered by \emph{nice slices that all contain} $R$ is smaller than $M- I(R)$. We will now argue that this implies that the null generators of both $H^+(B)$ and $H^-(B)$ are incomplete in the semiclassically allowed spacetime $M_{\slashed{R}}$, as predicted by Theorem~\ref{qst}. In particular, we demonstrate that \emph{any} Cauchy slicing of the spacetime $M$ with slices that contain $R$, the slices that probe the region beyond the endpoints of $H^{\pm}(B)$ necessarily have exponentially large extrinsic curvature.

\begin{figure}
\begin{center}
\includegraphics[width=0.9\textwidth]{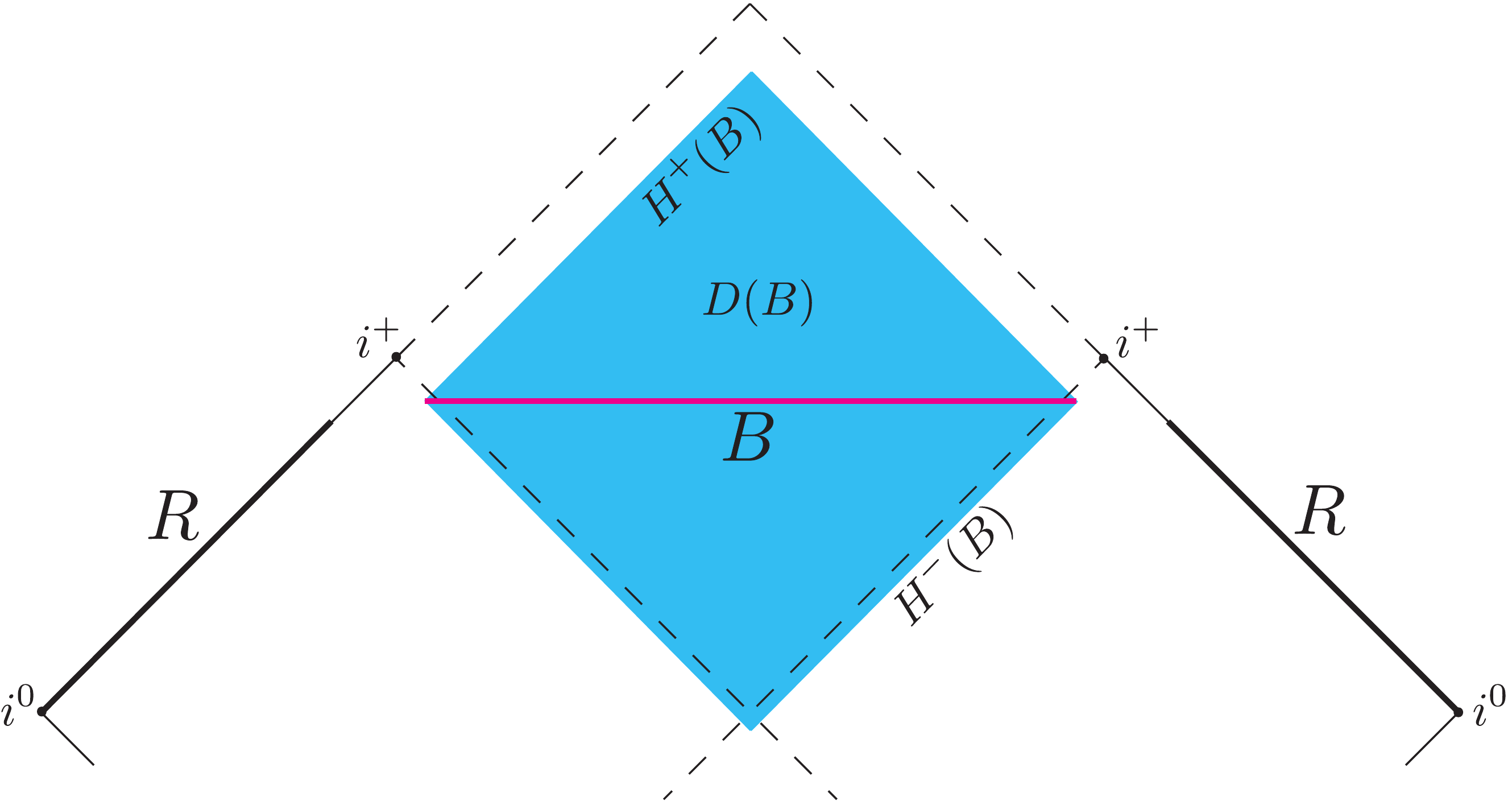}
\caption{\label{fig-RN1}Reissner–Nordstrom black hole. The inner and outer horizons are shown as dashed lines. In the Hartle-Hawking state, the conditions of the singularity theorem can be satisfied if $R$ is a large enough region near null infinity containing Hawking radiation and $B$ the interior of the black hole containing the purification of the radiation. This is puzzling since all null generators of both Cauchy horizons $H^\pm(B)$ contain their endpoints in the classical spacetime. However, semiclassically, $H^-$ encounters a quantum singularity when $R$ is fixed; $H^+$ does too, or else quantum corrections significantly alter the geometry near the inner horizon.}
\end{center}
\end{figure}

In this section, we speculate on why the semiclassical spacetime $M_{\slashed{R}}$ might be smaller than $M- I(R)$ leading to the incompleteness of the generators of both $H^+(B)$ and $H^-(B)$, upholding Theorem~\ref{qst}. We demonstrate that \emph{any} Cauchy slicing of the spacetime $M$ with slices that contain $R$, the slices that probe the region beyond the endpoints of $H^{\pm}(B)$ necessarily have exponentially large extrinsic curvature.

\subsection{Quantum Singularity On a Classically Regular Horizon}
\label{pastofb}

For concreteness, consider a maximally extended Schwarzschild black hole of radius $r_S$ in the Hartle-Hawking state, in 3+1 spacetime dimensions. Advanced and retarded time are defined by $u=t-r^*$ and $v=t+r^*$, where $r^* = r+r_S \log|(r/r_S)-1|$. The near horizon zone is the region $r_S<r<3r_S/2$; its outer boundary will be denoted $Z$. Below, we will also use Kruskal coordinates, $U = -2r_S e^{1-u/2r_S}$ and $V = 2r_S e^{1+v/2r_S}$, which cover the entire spacetime. (These are slightly nonstandard to match the standard Rindler coordinates.) We define $T=(U+V)/2$ and $X=(U-V)/2$.

Let $R$ be the union of a right asymptotic bulk region and its left mirror image; see Fig.~\ref{fig-twosidedpast}. On the right, $R$ is given by the portion $U<-U_0$ of a constant $t$ slice, with $t$ chosen large enough for $R$ to be far from the black hole. We choose $U_0$ past the Page time, that is, $U_0 \lesssim r_S e^{-\gamma_1 S}$ where $S$ is the Bekenstein-Hawking entropy of the black hole and $\gamma_1 \sim O(1)$. The boundary of the past of $R$ intersects the boundary of the near-horizon zone at $v_{\dot I(R)\cap Z}$. Choosing $B$ to be the black hole interior at the same (or slightly earlier) value of $v$, Theorem~\ref{qst} predicts a singularity along $H^-(B)$. 

Since the spacetime is classically regular in the past of $B$, this must be a quantum singularity. We will now verify this prediction.

\begin{figure}
\begin{center}
\includegraphics[width=0.9\textwidth]{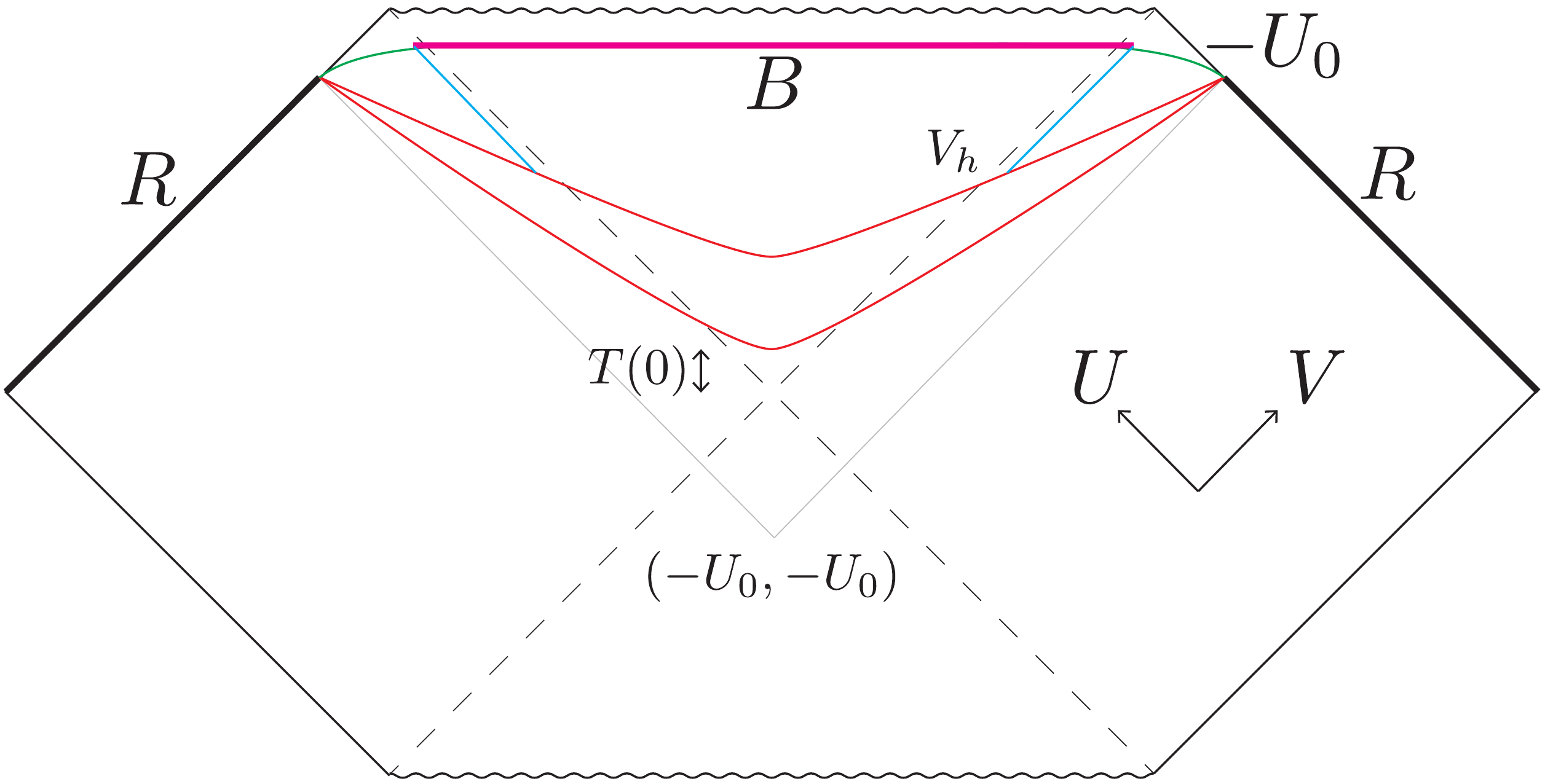}
\caption{\label{fig-twosidedpast}Schwarzschild black hole. In the Hartle-Hawking state, $B$ and $R$ can be chosen such that the conditions of the theorem can be satisfied towards the past of $B$. This is a quantum singularity: Cauchy slices containing $R$ which intersect the horizon around a Page time in the past of $B$ (red lines) have exponentially large extrinsic curvature.}
\end{center}
\end{figure}

For $v<v_{\dot I(R)\cap Z}$, $\dot I^-(R)$ lies within the Rindler region where the metric is well-approximated by
\begin{align}
   ds^2 = -dT^2 + dX^2 + r_S^2 d\Omega^2 + O(T^2, X^2)
\end{align}

Intuitively, this leaves little room between $\dot I(R)$ for spacelike slices that contain $R$ and enter the horizon very early. We will now argue that there exist no nice slices containing $R$ that intersect the horizon at or earlier than
\begin{equation}\label{eq-earlyv}
    v_0=v_{\dot I(R)\cap Z} - \gamma_2\, t_{\rm Page}~,
\end{equation}
where $\gamma_2\sim O(1)$. 

One of the necessary conditions for niceness is that the trace of the extrinsic curvature, $K$, is not too large. The precise condition is not clear to us. The early literature on nice slices~\cite{Lowe:1995ac,Polchinski:1995ta} suggests $|K| \ll 1/l_P$, but this may be too stringent. (In four or more spacetime dimensions, it would exclude slices that contain the Hawking radiation $R$ and its entanglement island.) We shall use the more lenient necessary condition
\begin{equation}\label{eq-kpoly}
    |K| \ll \frac{1}{l_P}\left(\frac{\ell}{l_P}\right)^n
\end{equation}
where $n>0$ is an unknown fixed constant, and $\ell$ is a characteristic scale of the geometry (here, $\ell\sim r_S$). We will see that even this rather weak niceness condition cannot be satisfied for any $n$, because $|K|$ becomes exponentially large at early times. 

Let $\Sigma$ be a left-right symmetric smooth Cauchy slice. (Thus, we assume that niceness cannot be rescued by using a slicing that spontaneously breaks the left-right symmetry.) $\Sigma$ is fully determined by a function $T= T_\Sigma(X)$.
$\dot{I}(R)$ is given by $T=|X|-U_0$. For $\Sigma$ to contain $R$ we must have $T_\Sigma > |X|-U_0$.
In the Rindler region, the extrinsic curvature is well approximated by
\begin{align}\label{eq-extrinsic}
    K_\Sigma = \frac{T''}{(1-T'^2)^{3/2}}~.
\end{align}
So long as $\Sigma$ is spacelike, this quantity is real.

Let $\Sigma$ intersect the horizon at $(U=0, V_h)$. The timelike proper distance between the intersection and $(-U_0, -U_0)$ is given by $\sqrt{(V_h+U_0)U_0}$. The smaller this distance, the larger $|K_\Sigma|$ needs to be if $\Sigma$ is not to become spacelike separated from $R$. To see this, first consider the special case where $T_\Sigma(0)$ is small enough to be in the Rindler region, with $T'(0)$ arbitrary. Given an upper bound $K_{\text{max}}$ on the magnitude of the extrinsic curvature, one finds for $|X|\gg 1/K_{\text{max}}$:
\begin{align}\label{eq-sigmaineq}
     T_\Sigma(X) < T_\Sigma(0) -\frac{1}{\sqrt{1-T'(0)^2}}\frac{1}{K_\text{max}}[1-T'(0) \mbox{sgn}(X)]+ |X|~.
\end{align}
Now consider the situation of interest: a slice $\Sigma$ which intersects the horizon at $(U=0, V_h)$ with $U_0 \lesssim r_S e^{-\gamma_1 S}$ and $V_h\lesssim \frac{r_S}{U_0}e^{-\gamma_2 S}$. (This corresponds to $u_0$ past the Page time and $v_h\leq v_0$ with $v_0$ given by Eq. \eqref{eq-earlyv}.) In the Rindler region, the sphere $(U=0, V_h)$ is related by a boost to the sphere $(X=0, T(0))$ with $T(0)=-U_0+\sqrt{(V_h+U_0)U_0}$. By Eq.~\eqref{eq-sigmaineq}, there exists no solution for $\Sigma$ with subexponential extrinsic curvature.

The only alternative to the presence of a quantum singularity is that the theorem fails, which means that one of its assumptions must fail. Indeed, if there were no restriction on the extrinsic curvature on a semiclassical slice, then the example in this subsection could be viewed as a violation of the quantum Bousso bound. Moreover, in the special case where $B$ is precisely the black hole interior, the example would furnish a violation of the Generalized Second Law of thermodynamics. Our viewpoint is that these assumptions are valid in the semiclassical regime, and that the theorem has simply uncovered a (possibly surprising) limitation of the semiclassical regime.

In simple models models where the entropy is approximated by a two dimensional CFT (see e.g.~\cite{Penington:2019npb, Almheiri:2019psf, Almheiri:2019yqk, Hartman:2020swn}), it is easy to show that the quantum focusing conjecture is satisfied along $H^-(B)$, even if we ignore the restriction to nice slices. In particular, the quantum expansion formally exists and remains negative along $H^-(B)$. However, the quantum Bousso bound is still violated by the (non-nice) slice that stays below the black hole. It is important to emphasize that starting with $\S$ on a nice slice, and then integrating the quantum expansion, is a valid method for computing the generalized entropy of any other \emph{nice} slice. In particular one is permitted to continue past caustics and self-intersections. Such features are generic, so this is a crucial ingredient in important semiclassical generalization of classical theorems. The problem in the present example is different: the slices do not stay nice. 

If one ignores this limitation, integration of $\Theta^{B\cup R}_-$ formally yields negative values of $\S$ well before the tip of the event horizon is reached, followed by a discontinuity when the slices no longer intersect $H^-(B)$. (To see this for an evaporating black hole, note that the quantum expansion along the horizon is to a good approximation independent of whether slices end at spatial infinity or at null infinity, and hence, so is the integrated change in $\S$. But the latter can be much greater than the generalized entropy of the complement of $B\cup R$ when $R$ is the Hawking radiation sufficiently far past the Page time.)  $\S<0$ has no interpretation as a von Neumann entropy in a fundamental theory, so this would be a nonsensical conclusion.

One might be tempted to ``save'' the GSL and the quantum Bousso bound for non-nice slices, by observing that the exact von Neumann entropy of $R$ receives non-perturbative corrections, which cause it to be bounded above by the Bekenstein-Hawking entropy of the black hole. This is not correct. 

First, the quantum Bousso bound is a semiclassical bound and is expected to apply to the semiclassical state, not to the nonperturbatively correct state. The same is true for the generalized second law. Calculating the generalized entropy outside of an evaporating black hole using the exact von Neumann entropy of radiation results in the violation of the generalized second law after the Page time, but in the semiclassical state the entropy of radiation increases throughout the process of evaporation, upholding the generalized second law.

Secondly, the conditions of our theorem can be satisfied even when there is no difference between the semiclassical and exact von Neumann entropy of $R$. For example, take $U_0$ to correspond to a few scrambling times, rather than the Page time. Then there exists a nonminimal quantum extremal surface associated to $R$. Now, consider moving this surface in the outward past null direction towards $\dot I(R)$. In simple 1+1 models with CFT matter, one finds that the generalized entropy of the enclosed region union $R$ decreases without bound. Therefore, at some point along the deformation the regions become hyperentangled. 
Furthermore, one can check that the quantum expansion also has the correct sign needed for the singularity theorem. However, again any Cauchy slice containing $R$ and dipping below the past tip of the event horizon necessarily has exponential extrinsic curvature. We view this as additional evidence for our nice slice criterion. (In fact, the above construction fails to yield a region that satisfies the assumptions of our theorem on a nice slice.)

\subsection{Classical vs.\ Quantum Singularity in Kerr-Newman Black Holes}
\label{inner}

Here we will discuss the singularity theorem applied to the future of $B$. In the Schwarzschild solution, the generators of $H^+(B)$ are obviously incomplete due to the curvature singularity at $r=0$, validating the prediction of our theorem. In charged or rotating black holes, however, the generators of $H^+(B)$ contain their endpoints, apparently violating Theorem~\ref{qst}. For simplicity, we will discuss this in detail for the Hartle-Hawking state of the maximally extended Reissner-Nordstrom black holes of nonzero charge, though we expect the main lessons to generalize to Kerr and Kerr-Newman black holes. The metric is given by:
\begin{align}\label{eq-RNmetric}
ds^2 = -f(r) dt^2 + f(r)^{-1} dr^2 + r^{2} d \Omega_{d-1}^2~,
\end{align}
where
\begin{align}
f(r) = \left(1-\frac{r_+}{r}\right) \left(1-\frac{r_-}{r}\right)~.
\end{align}
We pick $R$ to be the union of the asymptotic region similar to subsection~\ref{pastofb} and $B$ a late time slice of the interior such that $B$ is hyperentangled with $R$. Furthermore, the quantum expansion of $B \cup R$ along $H^+(B)$ can be easily arranged to be negative since the area variation towards the interior is large and negative. Therefore, it is easy to satisfy the conditions of our theorem towards $H^+(B)$. See Fig. \ref{fig-RN1}.

In the classical Reissner-Nordstrom background, the generators of $H^+(B)$ contain their endpoints (see Fig.~\ref{fig-RN2}). These lie on a sphere $\mu$ near the inner horizon bifurcation surface, of radius $r_-+\delta r$. The region $B$ is located around the Page time at the earliest, which implies
\begin{align}
    \delta r \lesssim L \exp{\left(-\frac{\alpha r_+^4}{r_-^2 G}\right)}~,
\end{align}
where $L$ is some function of $r_+$ and $r_-$, and $\alpha$ is an order one coefficient. Therefore, $\mu$ is exponentially close to the inner horizon bifurcation surface.

\begin{figure}
\begin{center}
\includegraphics[width=0.8\textwidth]{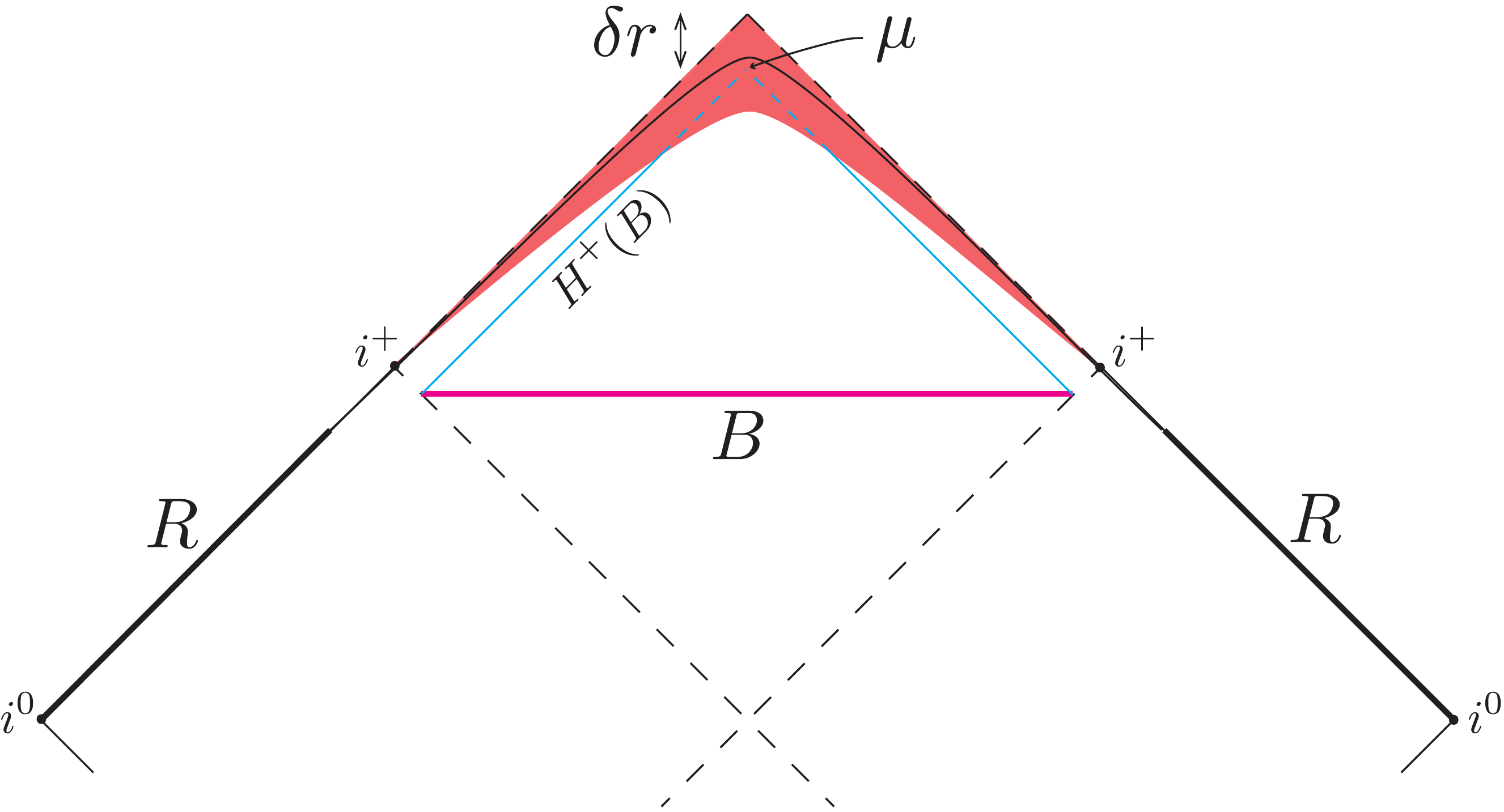}
\caption{\label{fig-RN2}In the Reissner–Nordstrom geometry, the region which is exponentially close to the inner horizon in area radius (shown in red) has the property that any Cauchy slice intersecting it necessarily has exponentially large extrinsic curvature somewhere. In the classical geometry, the generators of $H^+(B)$ come to an end on a sphere $\mu$ inside the red region. Therefore, a Cauchy slice (shown in black) which is nowhere to the past of $H^+(B)$ has exponentially large extrinsic curvature.}
\end{center}
\end{figure}

We do not expect that any nice slice containing $R$ reaches the future of $\mu$. As evidence for this, consider a constant $r$ slice which crosses $\mu$ or its future. Its extrinsic curvature satisfies
\begin{align}
    K \gtrsim \left[\frac{\left(\frac{r_+}{r_-}-1\right)}{L~r_-}\right]^{\frac{1}{2}} \exp{\left(\frac{\alpha r_+^4}{2r_-^2 G}\right)}~.
\end{align}

There is a second, seemingly independent reason why the semiclassical geometry may not contain $\mu$, and hence will satisfy the prediction of a singularity by Theorem~\ref{qst}. Quantum corrections to the matter stress tensor are known to become important near the inner horizon for a generic set of black hole parameters~\cite{Birrell:1978th, Sela:2018xko, Hollands:2019whz, Zilberman:2019buh, Zilberman:2022aum}.\footnote{For rotating BTZ, the stress tensor at leading order can be regular at the inner horizon~\cite{Dias:2019ery}. However, it has been argued that subleading corrections lead to a divergence~\cite{Emparan:2020rnp}.} For a simple toy model in which this can be shown, consider conformal matter in a 1+1 dimensional Reissner-Nordstrom background. The metric is
\begin{align}\label{eq-2dRN}
    ds^2 = -f(r)\, du\, dv~,
\end{align}
with $u=t-r^*$ and $v=t+r^*$ where $dr^* = dr/f(r)$. Setting the infalling flux to zero, the trace anomaly and conservation of the stress tensor imply:
\begin{align}\label{eq-RNT}
\langle T_{\mu\nu}k^\mu k^\nu \rangle \sim c \frac{\kappa_{-}^2-\kappa_{+}^2}{r^2}~,
\end{align}
where $k^\mu=\partial_v$, $\kappa_{+}$ and $\kappa_{-}$ denote the outer and inner horizon surface gravities and $c$ denotes the central charge of the CFT. Here we have added powers of $r$ by dimensional analysis to turn \eqref{eq-RNT} into an equation for 3+1 dimensions. A detailed derivation of the stress tensor can be found in \cite{Zilberman:2019buh} where it is shown that the coefficients in Eq.~\eqref{eq-RNT} are more complicated. We assume here that Eq.~\eqref{eq-RNT} is valid qualitatively for all $d$, though the detailed dependence on $r_\pm$ may be more complicated.

The stress tensor given in Eq.~\eqref{eq-RNT} seems regular because in $(u,v)$ coordinates the inner horizon is at infinity. In coordinates which are regular in a neighborhood of the inner horizon, the stress tensor can be shown to diverge. Here it is important to show that Eq.~\eqref{eq-RNT} will cause a large deviation of the metric from the classical geometry in a neighborhood of the inner horizon. The location $r$ at which the geometry gets $O(1)$ corrections from this stress tensor can be estimated by inspecting Raychaudhuri's equation for a spherically symmetric congruence near the inner horizon:
\begin{align}
\partial_v \theta_v = \kappa_v \theta_{v}-\frac{\theta_v^2}{d-1} - 8 \pi G  \langle T_{vv} \rangle
\end{align}
where $\kappa_v$ is the inaffinity. The quantum stress tensor becomes comparable to the other terms for $r-r_-\lesssim r_{\text{max}}$ where
\begin{align}
    r_{\text{max}}-r_-\sim \sqrt{G}
\end{align}
Therefore, it is clear that the generators of $H^+(B)$ exit the region which is well-approximated by the classical solution \eqref{eq-RNmetric}. To further understand the nature of the incompleteness of $H^+(B)$ would require knowledge of the correct geometry which is beyond the scope of this work. A natural guess would be that the geometry terminates at a spacelike singularity, directly upholding theorem~\ref{qst}. See Fig. \ref{fig-RN3_4}.

\begin{figure}
\begin{center}
\includegraphics[width=0.85\textwidth]{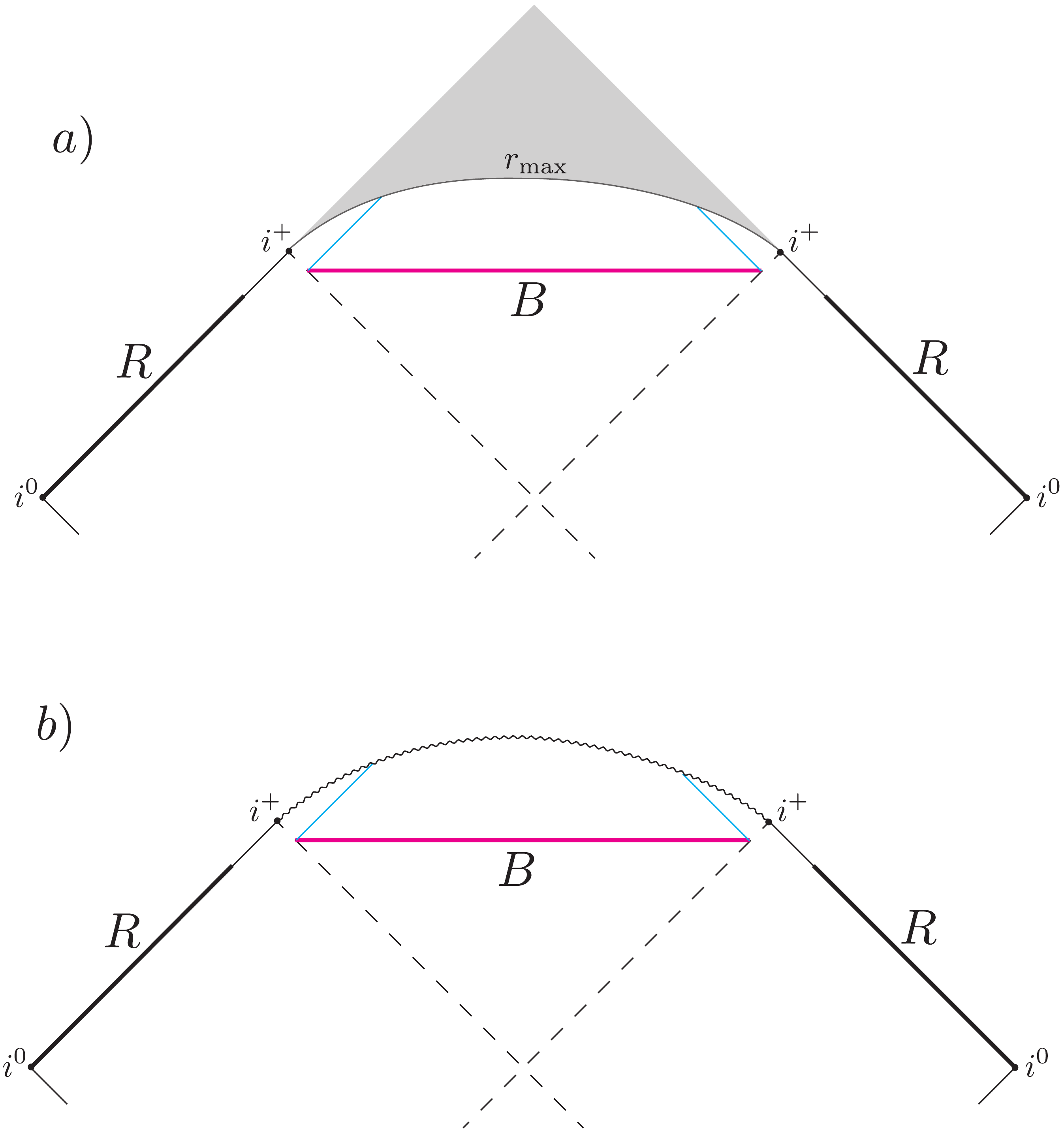}
\caption{\label{fig-RN3_4}(a) Generically the expectation value of the quantum stress tensor diverges near the inner horizon, causing large backreaction (grey).
The generators of $H^+(B)$ encounter this region, and our singularity theorem predicts that they are incomplete. The quantum-corrected geometry is not known, but the singularity theorem predicts that it either fails to admit nice slices or has a curvature singularity as shown in (b).}
\end{center}
\end{figure}

Our analysis has shown that quantum singularities blurs the line between an ordinary singularities and a Cauchy horizon. Even if quantum backreaction does not create an ordinary singularity, a quantum singularity forms before a Cauchy horizon can be reached. This may have some bearing on the strong cosmic censorship hypothesis, that physically reasonable spacetimes are globally hyperbolic. Strong cosmic censorship appears to be violated in the final stages of black hole evaporation and of the Gregory-Laflamme instability. These violations are in some sense small~\cite{Emparan:2020vyf} and should perhaps be ignored. The classical Kerr-Newman solution is of greater concern.
The fact that our theorem treats quantum singularities on the same footing as classical singularities encourages us to think of the evolution near the inner horizon as becoming singular, since no nice slices are available. We should treat this quantum singularity no differently than a classical one. Hence we need not rely on arguments that the backreaction from quantum effects would invalidate the classical Kerr-Newman solution near the inner horizon. Let $N$ be the past neighborhood of the inner horizon that cannot be reached by nice global Cauchy slices. We should treat $N$ the same as the small past neighborhood $N$ of the Schwarzschild singularity in which scalar curvature invariants approach or exceed the Planck scale: $N$ should not be part of the physical spacetime, the semiclassical geometric description terminates at the past boundary of $N$, and any geometric extensions of $M-N$ that we could consider are physically meaningless.

\section*{Acknowledgements}
We are grateful to Aidan Chatwin-Davies, Chitraang Murdia, Andrew Rolph, and Edgar Shaghoulian for their comments on Ref.~\cite{Bousso:2021sji}, which encouraged the present study. We also thank Ven Chandrasekaran, Tom Hartman, Adam Levine, Raghu Mahajan, Leonard Susskind, Marija Tomasevic, Aron Wall, Zhenbin Yang, and especially Juan Maldacena for useful discussions. We further thank Alexander Frenkel, Andrew Rolph, and Michelle Xu for comments on the first version of this paper. This work was supported in part by the Berkeley Center for Theoretical Physics; by the Department of Energy, Office of Science, Office of High Energy Physics under QuantISED Award DE-SC0019380 and under contract DE-AC02-05CH11231; and by the National Science Foundation under Award Numbers 2112880 (RB) and 2014215 (ASM).

\bibliographystyle{JHEP}
\bibliography{main}
\end{document}